%% file: main.tex
\documentclass[runningheads]{llncs}

\newif\ifEditMode

\EditModetrue
\usepackage{aliases}
\usepackage{preamble}

\begin{document}
\sloppy

\title{Liquidity Fragmentation or Optimization? \\Analyzing Automated Market Makers Across Ethereum and Rollups}

\titlerunning{Automated Market Makers Across Ethereum and Rollups} %TODO optional, please use if title is longer than one line

\author{Krzysztof M. Gogol\inst{1} \and
Manvir Schneider \inst{2}\and
Claudio J. Tessone \inst{1,3}\and \newline
Benjamin Livshits \inst{4}}

\authorrunning{Gogol et al.}

\institute{University of Zurich \and Cardano Foundation \and UZH Blockchain Center \and
Imperial College London}

\maketitle
\begin{abstract}
        \input{sections/00Abstract.tex} \\
%            \noindent\textbf{Keywords:} Automated Market Maker, Decentralized Finance, Rollup
%    \noindent\textbf{Keywords:} Automated Market Makers, Rollup, Liquidity Provider
\end{abstract}

%------------------------------------------------------------------------------

\input{sections/01Introduction}
\input{sections/02Background}

\input{sections/03Model}
\input{sections/04EmpiricalAnalysis}

\input{sections/08Discussion}
\input{sections/09Conclusions}

%------------------------------------------------------------------------------
\section*{Acknowledgments} 
The authors express their gratitude to Johnnatan Messias, Malte Schlosser, Aviv Yaish and Xin Wan for their insights and feedback during the course of this research.

Portions of this research were conducted by the first and fourth authors during their tenure at Matter Labs. This research article is a work of scholarship and reflects the authors' own views and opinions. It does not necessarily reflect the views or opinions of any other person or organization, including the authors' employer. Readers should not rely on this article for making strategic or commercial decisions, and the authors are not responsible for any losses that may result from such use.

%\newpage
\bibliographystyle{splncs04}
\bibliography{main}
	
\appendix
\input{sections/98Proofs}

\end{document}

%% file: sections/00Abstract.tex
Layer-2 (L2) blockchains inherit Ethereum’s security guarantees while reducing gas fees. As a result, they are gaining traction among traders at Automated Market Makers (AMMs), sparking debate over whether they contribute to liquidity fragmentation of Ethereum.
Our research suggests that such fragmentation is not currently occurring. However, it could emerge in the future—particularly if Liquidity Providers (LPs) recognize the higher returns available on L2s. Using Lagrangian optimization, we develop a model for optimal liquidity allocation across AMMs on Ethereum and its L2s, using staking as a benchmark. We show that, in equilibrium, AMM liquidity provision returns converge to this reference rate.
Additionally, we measure the elasticity of trading volume with respect to Total Value Locked (TVL) in AMMs and find that, on well-established blockchains, an increase in TVL does not necessarily lead to higher trading volume. Finally, our empirical findings reveal that Ethereum’s liquidity pools are oversubscribed compared to those on L2s and often yield lower returns than staking Ether. LPs could maximize their rewards by reallocating more than two-thirds of their liquidity to L2s and staking. 
%Layer-2 (L2) blockchains offer security guarantees of Ethereum while reducing transaction (gas) fees. Consequently, they are gaining popularity among traders at Automated Market Makers (AMMs), but Liquidity Providers (LPs) are lagging behind. By employing Lagrangian optimization, we show that, in equilibrium, liquidity provisions to AMMs should provide returns equal to risk-free rate across all blockchains, and we find the optimal liquidity allocation strategy that maximizes LP's rewards. Our empirical results show that liquidity pools on Ethereum are oversubscribed compared to their counterparties on L2s and deliver lower returns than staking Ether. LPs would maximize their rewards by reallocating over 2/3 of the liquidity to L2s and staking.  Lastly, we measure the elasticity of trading volume with respect to Total Value Locked (TVL) at AMM pools and found that at the well established blockchains an increase in TVL is not associated with an increase in trading volume.

%% file: sections/01Introduction.tex
%%%%%%%%%%%%%%%%%%%%%%%%%%%%%%%%%%%%%%%%%%%%%%%%%%%%%%%%%%%%%%%%%%%%%%%%%%%%%%%%%%%%%%%%%%%%%%%%%%%%%%%%%%%%%%%%%%
\section{Introduction}
    \label{sec:intro}
    
Automated Market Makers (AMMs)~\cite{Xu2021SoK:Protocols}, pioneered by Uniswap on Ethereum in 2018~\cite{Adams2020UniswapCore}, are the foundation of Decentralized Exchanges (DEXs). DEXs enable token exchanges without counterparty risk in an atomic blockchain transaction~\cite{Gogol2023SoKDeFi}. They were originally introduced to avoid the inefficiencies of the on-chain order book associated with high on-chain storage costs and security vulnerabilities~\cite{Jensen2021DeFi}. At AMM-DEX, exchange rates are set by a conservation function and current token reserves in the AMM pool. Liquidity Providers (LPs) supply tokens, earning trading fees paid by traders. As a result, a higher trading volume results in increased fees, which are subsequently distributed among LPs. Furthermore, a higher Total Value Locked (TVL) in liquidity pools attracts traders by reducing price impact during swaps, a primary cost component incurred during swap transactions~\cite{adams2024costs}.

Nevertheless, it is not always the optimal strategy for LPs to allocate their liquidity to the AMM pools with the highest trading volume. The trading fees accrued from traders are distributed among LPs in proportion to their respective contributions to the pool's TVL. Therefore, it is the ratio of fees to TVL, often referred to as capital efficiency, that should be the primary consideration for LPs.
LPs must also carefully select token pairs in the AMM pools. Volatile crypto-currency pairs lead to impermanent loss, reducing LP rewards~\cite{Heimbach2022RisksProviders}. Arbitrageurs equalizing token prices across AMM-DEXs, or with Centralized Exchanges (CEXs) amplify LPs' losses from market movement~\cite{milionis2023automated}. The Loss-Versus-Rebalancing (LVR) metric measures LP rewards relative to arbitrageurs~\cite{milionis2023automated}. Finally, should the rewards from liquidity provisions fall below those achieved by staking ETH, which is widely regarded the safest token allocation to yield rewards~\cite{gogol2024liquid}, the rationale behind the liquidity provision becomes questionable. The Loss-Versus-Staking metric compares LP returns to staking rewards~\cite{gogol2024liquid}.

The advent of rollups, a Layer-2 (L2) scaling solution~\cite{Thibault2022SoKRollup}, has shifted DeFi activities from Ethereum to its rollups. The success of rollups lies in their data integrity, which is secured by staked ETH, and the significant reduction in gas fees achieved by offloading computations off-chain~\cite{sguanci2021layer2,gangwal2022survey}. After the Ethereum Dencun upgrade in March 2024, the swap fees on L2s have dropped below \$0.01 per transaction. Consequently, the volume of swap transactions on rollups has surpassed that on the Ethereum mainnet, albeit with lower trading volumes. The reduced gas fees on rollups have facilitated transactions within the \$1-\$10 range, a previously unfeasible scenario on Ethereum. Furthermore, most rollups support the Ethereum Virtual Machine (EVM), allowing for the seamless deployment of AMM-DEXs originally designed for Ethereum onto rollups. Currently, DEXs such as Uniswap, Curve, and their forks have been deployed on rollups, enhancing the diversity of AMM pools accessible to LPs. As a result, LPs have more options to consider for providing liquidity, which serves as the starting point for this research.

Faster block production in rollups, along with cheaper gas fees, impacts the strategies of LPs in L2s. While Ethereum produces blocks every 12 seconds, rollups do it in 0.2 to 2 seconds depending on their implementation. The cheaper gas fees and faster block production allow LPs to adjust their strategies more often, which is especially vital for Concentrated Liquidity Market Makers (CLMM) such as Uniswap (v3)~\cite{adams2024layer2layer2}. With more frequency LP-position rebalancing in CLMMs, LPs earn higher rewards compared to similar pools on Ethereum.

This study investigates liquidity provision to AMM pools on L2s and Ethereum. After developing the theoretical framework for optimal liquidity allocation within the same cryptocurrency pools, we analyze our results using empirical on-chain data. 
%We assessed whether higher LP rewards on rollups, compared to those on Ethereum, can be linked to more efficient rebalancing of LP positions. These outcomes may stem from LPs' suboptimal liquidation allocation. 
We show that providing liquidity to certain AMM pools on rollups is more profitable compared to Ethereum. Despite the fact that liquidity pools on Ethereum have higher trading volumes, they often become oversubscribed, diminishing the rewards for LPs.

%define impermanet loss

%%%%%%%%%%%%%%%%%%%%%%%%%%%%%%%%%%%%%%%%%%%%%%%%%%%%%%%%%%%%%%%%%%%%%%%%%%%%%%%%%%%%%%%%%%%%%%%%%%%%%%%%%%%%%%%%%%
\parab{Related Work.} A comprehensive review of major AMM categories has been made available in~\cite{Xu2021SoK:Protocols} and expanded and updated in~\cite{Ruetschi2024DeFiTaxonomy}. The empirical study of liquidity provisions to AMMs has primarily focused on the Ethereum blockchain, especially for Uniswap (v2)~\cite{Heimbach2021BehaviorExchanges} and Uniswap (v3)~\cite{Fritsch2021ConcentratedMakers,Heimbach2022RisksProviders,Loesch2021ImpermanentV3,miori2024clustering,miori2023defi}. These investigations emphasize impermanent loss as a key indicator for assessing LP profitability across different cryptocurrency pairs. Furthermore, Tiruviluamala et al.~\cite{Tiruviluamala2022AMakers} proposed a framework to address impermanent loss. 

Alternative metrics for comparing LP profitability include LVR and LVS. Milionis et al.~\cite{milionis2023automated} compared LP rewards with arbitrage using loss-versus-rebalancing (LVR), while Fritsch~and~Canidio~\cite{fritsch2024profitability} extended its empirical analysis to more pools and showed that arbitrage profit increases on faster blockchains. Gogol et al.~\cite{gogol2024liquid} introduced loss-versus-staking (LVS) as another comparative measure, evaluating LP returns against staking rewards. Yaish et al.~\cite{Yaish2023Suboptimality} demonstrated the suboptimal behavior of LPs in DeFi lending pools.

Research on AMMs on L2s is nascent, focusing mostly on MEV and arbitrage.
Torres et al.~\cite{torres2024rolling} reported on cycling arbitrage (MEV) within L2s, while Gogol et al.~\cite{gogol2024quantifying,gogol2024l2arbitrage} examined non-atomic arbitrage involving cross-rollup and L2-CEX.
In his pioneering work, Adams~\cite{adams2024layer2layer2} observed that the liquidity concentration at Uniswap (v3) on Arbitrum and Optimism surpasses that on Ethereum by 75\%, indicating that LPs often readjust their positions on L2s. Chemaya~and~Liu~\cite{Chemaya2022Preferecnes} estimated AMM traders preferences for blockchain security on two main L2 networks in comparison to Ethereum.
%Our study builds upon these findings by analyzing the fees/TVL ratio and identifying an optimal LP allocation strategy among AMMs on Ethereum and its L2s.

%Write how profitability of LP can be measured - LVH, LVR, LVS. 

%%%%%%%%%%%%%%%%%%%%%%%%%%%%%%%%%%%%%%%%%%%%%%%%%%%%%%%%%%%%%%%%%%%%%%%%%%%%%%%%%%%%%%%%%%%%%%%%%%%%%%%%%%%%%%%%%%

\parab{Contribution.} This research analyzes both empirically and theoretically the liquidity provision for AMMs on L2s. We assess LP rewards on Uniswap (v3) within Ethereum and its optimistic rollups (Arbitrum, Base, Optimism), and ZK rollups (ZKsync). The contributions of this work are as follows:
%Each hypothesis is validated through the formulation of mathematical equations, followed by simulations and empirical analyses using historical on-chain data from Ethereum and its significant roll-ups: zkSync Era, Arbitrum, and Optimism. This research enhances the understanding of profitability in liquidity provisions at AMM-DEXs, particularly at Curve v2. We assess the profitability of liquidity pools in AMMs by determining the return from fees per unit of invested capital. Additionally, we compute the optimal maximum capital that should be transferred from one pool to another.

\begin{itemize}
    \item Using Lagrangian optimization, we find the optimal allocation of liquidity across staking and AMMs pools on Ethereum and its rollups, with the objective of maximizing LP rewards. We further show that in equilibrium, liquidity provisions to AMMs should provide returns equal to risk-free rate across all blockchains.
    \item We measure the elasticity of trading volume with respect to TVL at AMM pool and found that at the well established blockchains --- Ethereum, Arbitrum and Optimism, contrary to expectations, an increase in TVL is not associated with an increase in trading volume. In contrast, emerging blockchains, Base and ZKsync, exhibit a positive elasticity value, indicating that the volume of trade is positively correlated with TVL on these chains. 
    \item We empirically find that the current allocation of liquidity to WETH-USDC pools of Uniswap (v3) on Ethereum and rollups is not optimal for LPs. The pool on Ethereum tend to be overcapitalized and does not compensate LP for the missed opportunity to stake the entire capital. Specifically, over 66\% of Ethereum liquidity should be reallocated to rollups in order to maximize LP returns and to attain equilibrium with staking rates.
\end{itemize}

%%%%%%%%%%%%%%%%%%%%%%%%%%%%%%%%%%%%%%%%%%%%%%%%%%%%%%%%%%%%%%%%%%%%%%%%%%%%%%%%%%%%%%%%%%%%%%%%%%%%%%%%%%%%%%%%%%

\begin{comment}
\subsubsection{Paper Organization.} 
%Section~\ref{sec:background} provides an introduction to rollups and AMMs. 
Section~\ref{sec:model} presents the theoretical model for evaluating the profitability of liquidity provisions for AMM in various roll-ups. 
Section~\ref{sec:dataset} describes the historical data of the AMM pools that we sourced on-chain, followed by the empirical analysis in Section~\ref{sec:analysis}. 
The final sections, \ref{sec:discussion} and \ref{sec:conclusions}, discuss the results and summarize conclusions. 
\end{comment}

%% file: sections/02Background.tex
%%%%%%%%%%%%%%%%%%%%%%%%%%%%%%%%%%%%%%%%%%%%%%%%%%%%%%%%%%%%%%%%%%%%%%%%%%%%%%%%%%%%%%%%%%%%%%%%%%%%%%%%%%%%%%%%%%
\section{Background}
    \label{sec:background}

This section presents Layer-2 (L2) blockchain scaling solutions, including roll-ups. Then it introduces Automated Market Makers (AMMs) and their major categories.

%%%%%%%%%%%%%%%%%%%%%%%%%%%%%%%%%%%%%%%%%%%%%%%%%%%%%%%%%%%%%%%%%%%%%%%%%%%%%%%%%%%%%%%%%%%%%%%%%%%%%%%%%%%%%%%%%%
\subsection{Layer-2 Blockchain Scaling}
According to the blockchain scalability trilemma~\cite{delmonte2020decentralization}, blockchains can prioritize only two of three: decentralization, security, or scalability. Ethereum, the main platform for DeFi with the highest TVL, prioritizes decentralization and security. This led to network congestion, high gas fees, and throughput limited to 12 transactions per second (TPS). The Layer 1 (L1) and Layer 2 (L2) scaling solutions address blockchain scalability. L1 scaling involves the development of new blockchains with novel consensus mechanisms~\cite{Lashkari@2021Consensus}, sharding~\cite{Wang@2019Sharding}, and their own physical infrastructure. In contrast, in L2 scaling intensive computations are executed off-chain, with their results being recorded on the underlying L1 blockchain~\cite{sguanci2021layer2,gangwal2022survey}. 

\parai{Rollup.} 
Rollups~\cite{Thibault@2022rollups}, the non-custodial form of L2 scaling, act as blockchains. They generate blocks, execute transactions, and subsequently record compressed batched data on the L1. This approach ensures that the integrity of the rollup data is guaranteed by L1 security, such as staked ETH in the case of Ethereum's rollups. In order to maliciously modify the rollup state history, an attacker must compromise the security of the underlying L1 network.

\parai{Sequencer.} 
A sequencer~\cite{chaliasos2024formalrollups} is an integral component of rollup, tasked with executing and ordering transactions, forming blocks, and creating batches that are uploaded to the L1 chain. By bundling transactions together, rollups provide more gas-efficiency compared to the L1 network. To avoid additional trust assumptions in sequencer operators, rollups use optimistic or zero-knowledge proofs (ZKPs) to ensure the correctness and integrity of the data. Fig.~\ref{fig:rollup} illustrates the architecture of optymistic and ZK-rollups. Presently, major Ethereum optimistic rollups (e.g., Arbitrum, Optimism, Base) and ZK-rollups (e.g., ZKsync, StarkNet) rely on centralized sequencers and, in the case of ZK-rollup, centralized provers. %Ongoing research aims to decentralize sequencers~\cite{motepalli2023soksequencers} and provers~\cite{wang2024mechanismprover} to improve the reliability and security of both types of rollups.

\parai{Optimistic Rollup.} 
Optimistic rollups~\cite{Kalodner@2018Arbitrum} operate on a trust-based system, assuming transactions are valid unless disputed. This approach simplifies the implementation of optimistic rollups, particularly in supporting the Ethereum Virtual Machine (EVM). Consequently, optimistic roll-ups were faster to launch EVM compatibility and attracted higher DeFi adoption. However, the optimistic fraud-proof mechanism can lead to delays in withdrawals. Currently, most optimistic rollups enforce a 7-day challenge period.

\parai{ZK-Rollup.} 
In contrast, ZK-rollups~\cite{Chaliasos@2024AnalyzingRollups} leverage ZKPs to validate the state on L1 immediately after a proof has been generated off-chain by \emph{provers} and submitted to \emph{verifiers}. Verifiers, smart contracts on L1, validate transactions aggregated by the sequencer and confirm their correctness. Consequently, rapid finality is ensured, albeit for increased computational costs. ZK-rollups also offer enhanced compression opportunities, e.g. by posting to L1 only ZKPs instead of all transaction data.

\begin{figure}[!ht]
  \centering
  \begin{tabular}{cc}  % Create a 2-column table for side-by-side figures
    % First subfigure
    \includegraphics[width=0.45\textwidth]{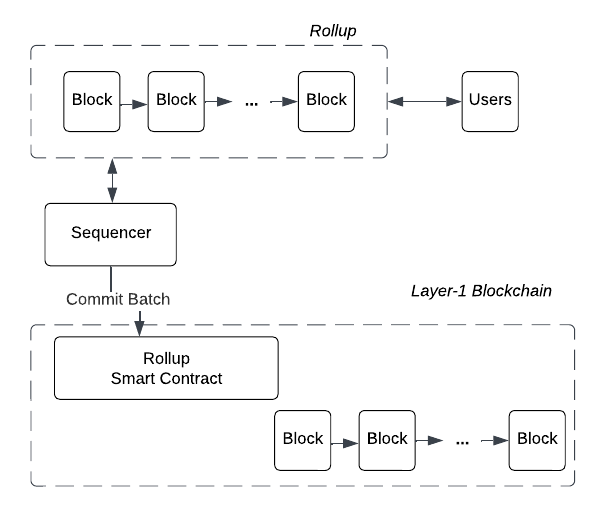} &
    % Second subfigure
    \includegraphics[width=0.45\textwidth]{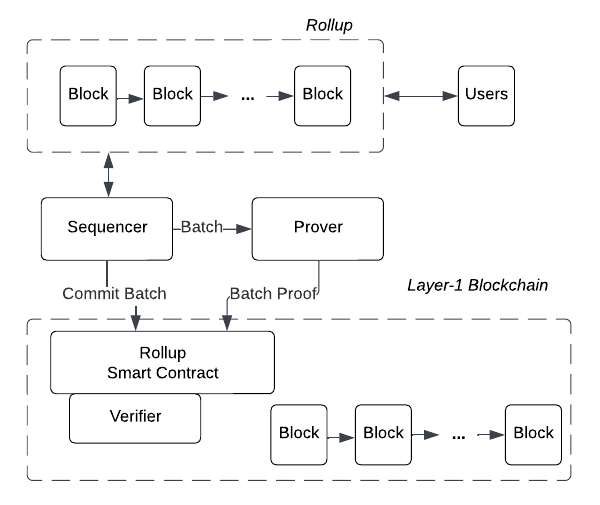} \\
    
    % Captions under each figure
    \small (a) Optimistic rollup & \small (b) ZK-rollup\\
  \end{tabular}

  \caption{High-level architecture of rollups}
  \label{fig:rollup}
\end{figure}

%%%%%%%%%%%%%%%%%%%%%%%%%%%%%%%%%%%%%%%%%%%%%%%%%%%%%%%%%%%%%%%%%%%%%%%%%%%%%%%%%%%%%%%%%%%%%%%%%%%%%%%%%%%%%%%%%%

\subsection{Automated Market Maker}
Automated Market Makers (AMMs) are blockchain protocols to facilitate token exchanges in a completely non-custodial manner. The primary elements of these protocols are liquidity pools comprising token reserves, which are contributed by \emph{Liquidity Providers} (LPs). \emph{Traders} (Swappers) transact with the AMM liquidity pool at exchange rates derived by the AMM, based on the token reserves. Each token swap incurs a trading fee (LP fee), calculated as a fixed percentage of the trade volume. These fees are allocated, in whole or in part, among the LPs within the pool on a proportional basis. \emph{Arbitrageurs} capitalize on price discrepancies between AMM-DEXs or between AMM-DEXs and CEXs to equilibrate market prices.

Constant Function Market Makers (CFMM) represent the main type of AMM and are characterized by a trading function, commonly referred to as a conservation function or reserve curve. This trading function maps token reserves to the invariant, which remains constant throughout each swap transaction. The number of tokens within the pool, $N$, is typically $N=2$. The token reserves within liquidity pools are represented by $x_1,\ldots,x_N$. A diverse array of CFMMs has been proposed \cite{Xu2021SoK:Protocols} to mitigate the trading cost (price impact) for traders and maximize capital efficiency for LPs.

%    \paragraph{Constant Sum} CSMM -- $\sum_{i=1}^N x_i = k$
    \paragraph{Constant Product (CPMM)}
  The first AMM, pioneered by Uniswap, is characterized by a straightforward invariant \cite{Adams2020UniswapCore}, which is still used for trading tokens of high volatility or those whose prices have yet to stabilized:
            $$ x_1 x_2 = L^2.$$
    \paragraph{Constant Product with Concentrated Liquidity (CLMM)}
        CLMM was introduced by Uniswap v3 \cite{Adams2021UniswapCore} to enhance the efficiency of TVL within the liquidity pool, allowing LPs to designate a specific price range, $[p_a, p_b]$, for their liquidity contributions. Consequently, the trading invariant within this specified interval is
            \begin{equation*}\label{eq:uniswapv3}
            \left( x_1+\frac{L}{\sqrt{p_b}} \right)\left(x_2+L\cdot \sqrt{p_a}\right) = L^2.
            \end{equation*}
    % \paragraph{Stableswap Invariant}
    %     A drawback of CLMM is the necessity for LPs to continuously monitor and adjust the designated price range. The Stableswap Invariant, introduced in Curve v1 \cite{Egorov2019StableSwapLiquidity}, autonomously centralizes liquidity around the market price through the utilization of the constant \(A\). This mechanism is tailored for tokens that consistently trade at a fixed price relative to one another, such as stablecoins anchored to 1 USD. The mathematical expression of the invariant is
    %     \begin{equation}\label{eq:curvev1}
    %         K\cdot D^{N-1}\cdot \sum\limits_{i=1}^N x_i +\prod\limits_{i=1}^N x_i = K\cdot D^N + \left(\frac{D}{N} \right)^N
    %     \end{equation}
    %     with
    %         $$K= \frac{A\cdot\prod_{i=1}x_i}{D^N}\cdot N^N,$$
    %     where $A$ is a parameter.
    % \paragraph{Cryptoswap Invariant}
    %     Curve v2 \cite{Egorov2021CurvePeg} represents a further modification of the Stableswap invariant capable of accommodating any token pairs. It employs the same mathematical formulation as equation (\ref{eq:curvev1}), with the variable $K$ being contingent upon the parameters $A$ and $\gamma$, as presented below
    %         \begin{equation}\label{eq:curvev2}
    %         K = A\cdot \underset{\text{ =: }K_0}{\underbrace{\frac{A\cdot\prod_{i=1}x_i}{D^N}\cdot N^N}} \cdot \frac{\gamma^2 }{(\gamma+1- K_0)^2}.\end{equation}

%% file: sections/03Model.tex
%%%%%%%%%%%%%%%%%%%%%%%%%%%%%%%%%%%%%%%%%%%%%%%%%%%%%%%%%%%%%%%%%%%%%%%%%%%%%%%%%%%%%%%%%%%%%%%%%%%%%%%%%%%%%%%%%%
\section{Model}
\label{sec:model}

Our model comprises $n$ blockchains: Ethereum and its rollups. Each blockchain has an identical AMM deployed with the liquidity pool of the same two cryptocurrencies and the same configuration.
\begin{itemize}
\item Each liquidity pool $i$ has a different $\text{TVL}_i$ and trading volume $\text{Vol}_i$. The trading fees $f$ are the same for each pool, and the total fees of the pool earned by LP are $fees_i := f\cdot \text{Vol}_i$.
\item Each liquidity pool $i$ consists of the same tokens and consequently has the same impermanent loss (IL).
\item Blockchains vary in gas fees and block production times.
\end{itemize}
In our model, we also consider that LPs have the possibility to stake and earn the staking rewards $r$. We consider the staking rate $r_s$ to be similar to the risk-free rate in traditional finance. Consequently, we assume that if liquidity pools do not provide higher rewards than the staking $r_s$, the LP reallocates her wealth to staking. Given that, within the Ethereum ecosystem, the volume of capital allocated to staking vastly exceeds that which is locked within liquidity pools~\cite{Gogol2023SoKLST}, we assume that the staking rate remains unaffected by staked capital in our model.

%There are three categories of DeFi users: traders, liquidity providers (LPs), and arbitrageurs. Arbitrageurs are a unique subset of traders who aim for nearly risk-free profits by taking advantage of price discrepancies across different DeFi protocols and centralized finance (CeFi), thereby balancing the prices. Additionally, there is a fourth category known as governance users.

\begin{comment}
\subsection{Trader}
We first build the model for the total costs of AMM swap for a trader.

\begin{align}
\label{eq:cost}
Total Swap Cost = L1 Fee + L2 Fee + LP Fee  + Price Impact + Slippage Cost
\end{align}
We ignore L1 and L2 fees and slippage as they do not depend on the swap volume. Let f be the LP fee charged on the volume. We assume for now that f is constant. In some AMMs like Curve, or upcoming Uniswap v4, f will be a function. We have:
$$Total Cost of Swap (Volume) = f * Volume  + Price Impact (Volume) $$
We know the formulas for price impact for Uniswap v2, v3, Curve v1 and v2. Let x and y be the amount of token reserves in AMM. For Constant Product Market Maker (CPMM) - Uniswap v2, we have:
$$Total Cost of Swap (Volume) = f * Volume  + \frac{Volume}{x} $$
\end{comment}

\subsection{Constant Product Market Maker}
Liquidity provider aims to optimize her earnings, which are directly related to the capital $w_{i}$ she contributes to the liquidity pool on blockchain $i$ with total value locked $\text{TVL}_i$ and total fees earned $fees_i$. The earnings for LP are calculated as follows: 
\begin{equation}
    \frac{w_i}{\text{TVL}_i + w_i} fees_i =  \frac{w_i}{\text{TVL}_i + w_i} f \cdot \text{Vol}_i
\end{equation}

%$$f \times Volume \frac{TVL_{LP}}{TVL}$$, 
%and $$f \times Volume \frac{x_{LP} + y_{LP}}{x + y \times S}$$.

\noindent
Thus, the return rate of LP on fees for allocating $w_i$ capital to the liquidity pool on blockchain $i$ is:

\begin{equation}\label{eq:return}
    r_i(w_i) = \frac{f \cdot \text{Vol}_i}{\text{TVL}_i + w_i}
\end{equation}

\noindent
Returns, as presented by DEX aggregators and the GUI of AMM-DEX, are $r_i(0)$. In order to calculate the LP return, impermanent loss (IL) must be deduced from $r_i(0)$.

\subsection{Concentrated Liquidity Market Maker (CLMM)}
Uniswap (v3) enables LPs to define the price range within which liquidity is supplied, thereby enhancing capital efficiency. Consequently, LPs accrue fees solely on the capital allocated to the specific tick where swaps occur. 
To compare the ETH-USDC pools on Uniswap (v3) across Ethereum and its rollups, each with its unique liqudity concentrations, we evaluate the profitability per unit of ambient (unbounded) liquidity and seek equilibrium conditions that equalize their returns.

Thus, the return of LP for allocating $w_i$ capital (unbounded) to the pool on the blockchain $i$ is the following:
\begin{equation}\label{eq:return_tick}
    r_i(w_i) = \frac{f \cdot \text{Vol}_i}{\text{TVL}_i^j + \frac{w_i}{m}}.
\end{equation}

\noindent
where $\text{TVL}_i^j$ is the liquidity in the current tick $j$ and $m$ is the number of ticks

%\subsection{Equilibrium Conditions.}
\subsection{Optimal Allocation}
Assume there is one LP. We denote by $w_0$ the amount of capital allocated to staking at the staking rate $r_s$ and for $i=1,\ldots,n$ we denote $w_i$ as the capital allocated to liquidity pool $i$. We are looking for a vector $\mathbf{w}=(w_0, w_1, \ldots, w_n)$ that maximizes LP's earnings. We assume that the LP has total liquidity $W$ to allocate to staking and AMM pools.
% \R{We consider two cases, i) when the market should be in equilibrium
% $\sum_{i=0}^{n}w_i = 0 $ and ii) when LP has liquidity to allocate to staking and AMM pools $\sum_{i=0}^{n}w_i = W > 0 $. In the course of the empirical analysis, we examine the dynamics of these vectors over time.}
By allocating $w_0$ to staking, the return from staking is $r_s w_0$.

% \textcolor{blue}{The return from staking $w_0$ is $r_s w_0$. If the return from each pools $i$ equals the return from staking, then there is no incentive for the LP to change to staking. Fix a $w_0$. For all $i>0$, we need $w_i r_i(w_i)=r_s w_0$. Using Equation~\eqref{eq:return} we need
% \begin{equation*}
%     w_i \frac{f \times \text{Vol}_i}{\text{TVL}_i + w_i} = r_s w_0
% \end{equation*}
% d s
% \begin{equation*}
%      \frac{f \times \text{Vol}_i}{\text{TVL}_i + w_i} = r_s 
% \end{equation*}
% d
% which we can rewrite to 
% \begin{equation}
%     w_i =  \frac{r_s w_0 \cdot \text{TVL}_i}{f \times \text{Vol}_i - r_s w_0}.
% \end{equation}
% The latter equation needs to hold for all $i>0$, which gives the following:
% Find $w_0$ s.t. $\left[ w_0 + \sum_{i=1}^n \frac{r_s w_0 \cdot \text{TVL}_i}{f \times \text{Vol}_i - r_s w_0} \right] = W$.
% }
Given an allocation vector $\mathbf{w}$, the total earnings for the LP are $r_s w_0 + \sum_{i=1}^n r_i(w_i)w_i$. The LP's objective is to maximize his earnings, that is, she faces the following optimization problem (using Equation~\eqref{eq:return}):
\begin{equation}\label{eq:optimizationProblem}
\max_{\mathbf{w} = (w_0,\ldots,w_n)}  r_s w_0 + \sum_{i=1}^{n} \frac{w_i}{\text{TVL}_i + w_i} f \cdot \text{Vol}_i
\end{equation}
subject to $w_i \geq 0$  for each $i = 0, \ldots, n$ and $\sum_{i=0}^{n} w_i = W$. 

Incorporating impermanent loss into the above model can easily be done by increasing the reference rate $r_s$ with the rate of impermanent loss, as each liquidity pool incurs an equal impermanent loss.

\begin{proposition}\label{prop: optallocation}
    The allocation vector $\mathbf{w}=(w_0,\ldots,w_n)$ with 
    \begin{align}\label{eq:prop1}
        w_i = \text{TVL}_i (\sqrt{r_i(0)r_s^{-1}} - 1),    \quad i=1,\ldots,n
    \end{align}
    and $w_0 = W - \sum_{i=1}^n w_i$ is the solution to \eqref{eq:optimizationProblem} and maximizes LP's earnings.
\end{proposition}

\begin{proof}
We define the Lagrangian $\mathcal{L}(\mathbf{w}, \lambda) = r_s w_0 + \sum_{i=1}^n \frac{f \cdot \text{Vol}_i}{\text{TVL}_i + w_i} w_i + \lambda \left( W - \sum_{i=0}^n w_i \right)$
where $w_i \geq 0$ and $\sum_{i=0}^{n} w_i = W$. Taking the derivatives with respect to $w_0$, $w_i$, and $\lambda$, and solving yields the optimal conditions (\ref{eq:prop1}). \qed
\end{proof}

%%%%%%%%%%%%%%%%%%%%%%%%%%%%%%%%%%%%%%%%%%%%%%%%%%%%%%%%%%%%%%%%%%%%%%%%%%%%%%%%%%%%%%%%%%%%%%%%%%%%%%%%%%%%%%%%%%
\subsection{Convergence}
Assume now that there are $m$ independent and identical LPs that invest into the staking and AMMs consecutively. Each action of an LP changes the return rates of the AMMs. Based on these changes other LPs will take action accordingly. In particular, an allocation of capital to a pool increases its TVL and thereby decreases the rate of return, see Equation~\eqref{eq:return}. If LPs allocate their capital consecutively, the rate of return of each pool will decrease until it reaches $r_s$. Once each pool reaches $r_s$, the remaining LPs will invest into staking. Formally, we have the following convergence result. Let $r^j_i$ be the rate of return for pool $i$, after $j-1$ LPs have invested their capital. Note that after each investment by an LP, the respective pool's TVLs increase by that amount.
\begin{proposition}\label{prop:convergence}
    Assuming $m$ LPs that consecutively invest into the staking and AMMs following \eqref{eq:prop1}. The rate of return from each pool converges to the staking rate $r_s$. In particular, $r_i^m(0) \overset{m \to \infty}{\to} r_s$ for all $i=1,\ldots,n$.
\end{proposition}

\begin{proof}
     LPs and let the LPs invest consecutively. The rate of return can be simplified using Proposition~\ref{prop: optallocation}. In particular, 
        $r_i^{j}(w^j_i) = \sqrt{ r_i^{j-1}(0) }  \sqrt{r_s}$
    Note that the latter is independent of $w_i^j$ and the convergence result follows $j=m$. \qed
\end{proof}

% Let's assume there are $m$ LPs: $LP_j$ that invest into the staking and AMMs subsequently.
% After the allocation of $LP_1$ the $\text{TVL}_i$ is  $\text{TVL}_i + w_i^1$ and the rate of return is:

% \[
% r_i^n(0) = \frac{f \cdot \text{Volumes}_i}{\text{TVL}_i^{n-1} + (\sqrt{\frac{f \cdot \text{Volumes}_i \cdot \text{TVL}_i^{n-1}}{r_s}} - \text{TVL}_i^{n-1})} 
% = \sqrt{ \frac{ f \cdot \text{Volumes}_i}{\text{TVL}_i^{n-1}} }  \sqrt{r_s}
% \]

% \[
% r_i^n(0) = 
% = \sqrt{ r_i^{n-1}(0) }  \sqrt{r_s} = (r_i^0)^{\frac{1}{2n}} * r_s^{\frac{1}{2} + \frac{1}{4} + ... + \frac{1}{2n}}
% \]
%From the equation above we can see that $r_i^n(0)$ converges to $r_s$ with n.

\begin{example}
    We illustrate in Fig.~\ref{fig:opt_allo} and Fig.~\ref{fig:ror} the convergence result from Prop.~\ref{prop:convergence} for one AMM and $m=8$ LPs. We assume a given wealth of $W=10'000'000$ per LP, a staking rate of $r_s=3.42\%$, a fee value $f=1$, $\text{Vol} = 400'000$ and $\text{TVL}=4'000'000$. The LPs sequentially allocate their funds to the AMM and/or to staking.
\begin{figure}[t]
  \centering
  \begin{minipage}{0.45\textwidth}
    \centering
    \includegraphics[width=\textwidth]{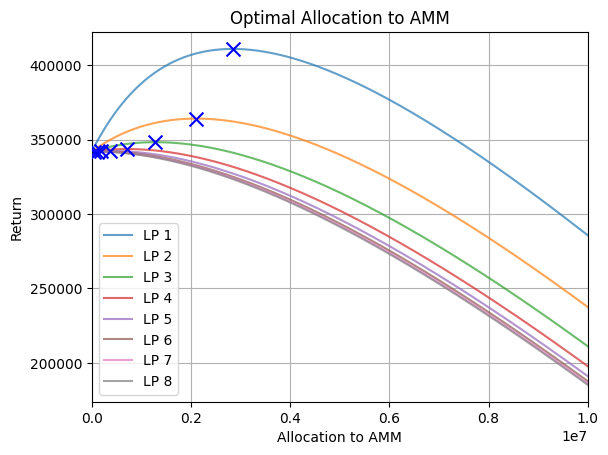}
    \caption{Optimal allocation for each LP (sequential allocation).}
    \label{fig:opt_allo}
  \end{minipage}
  \hfill
  \begin{minipage}{0.43\textwidth}
    \centering
    \includegraphics[width=\textwidth]{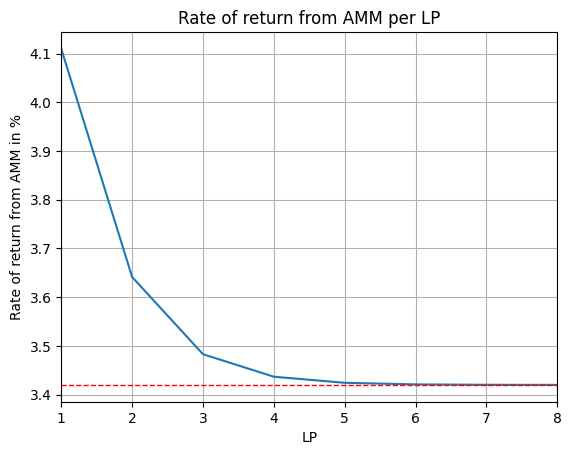}
    \caption{Rate of return of the AMM converges to the staking rate $r_s$}
    \label{fig:ror}
  \end{minipage}
\end{figure}

\end{example}

\subsection{Modeling Volume Dynamics}

We introduce a parameter for volume elasticity~\cite{2025Elasticity}, denoted by $\epsilon_v$. The elasticity parameter captures how sensitive trading activity is to changes in liquidity.
The trading volume follows the functional form 
\begin{equation}\label{eq: Vol}
\text{Vol} = k \cdot \text{TVL}^{\epsilon_v},
\end{equation}
where $k$ is a scaling constant. This form is particularly motivated because deeper liquidity pools (i.e. pools with higher liquidity) reduce price impact and slippage, incentivizing larger trades. 

\begin{itemize}
    \item $\epsilon_v > 1$: Highly elastic volume, where trading activity increases significantly as liquidity grows.
    \item $\epsilon_v \leq 1$: Inelastic volume, where trading volume grows slower than liquidity.
\end{itemize}
Following the previously introduced notation, the LP’s objective is to maximize their earnings:
\begin{equation}\label{eq:optimizationProblem_eps}
\max_{\mathbf{w} = (w_0,\ldots,w_n)}  r_s w_0 + \sum_{i=1}^{n} f \cdot k_i \cdot (\text{TVL}_i + w_i)^{\epsilon_v - 1} w_i
\end{equation}
subject to $w_i \geq 0$ for each $i = 0, \ldots, n$ and $\sum_{i=0}^{n} w_i = W$. Note that we replaced $\text{Vol}_i$ with \eqref{eq: Vol}, considering that the LP's contribution $w_i$ to pool $i$ must be added to $\text{TVL}_i$.

% \begin{proposition}
% The allocation vector $\mathbf{w} = (w_0, w_1, \ldots, w_n)$ with
% \begin{equation}\label{eq:optimal_allocation}
% w_i = \text{TVL}_i ( \sqrt{f \cdot k_i \cdot (\text{TVL}_i + w_i)^{\epsilon_v - 1}r_s^{-1}} - 1 ), \quad i = 1, \ldots, n,
% \end{equation}
% and $w_0 = W - \sum_{i=1}^{n} w_i,$
% is the solution to the optimization problem (\ref{eq:optimizationProblem_eps}) and maximizes LP's earnings.
% \end{proposition}

% \begin{proof}
% The statement is proved as in Proposition~\ref{prop: optallocation} using the Lagrangian $\mathcal{L}(\mathbf{w}, \lambda) = r_s w_0 + \sum_{i=1}^{n} \frac{f \cdot k_i \cdot (\text{TVL}_i + w_i)^{\epsilon_v - 1}}{\text{TVL}_i + w_i} w_i + \lambda \left(W - \sum_{i=0}^{n} w_i\right).$ \qed
% \end{proof}

\begin{proposition}
    The allocation vector $\mathbf{w} = (w_0, w_1, \ldots, w_n)$ which solves
    \begin{equation}\label{eq:optimal_allocation}
w_i = \left( \frac{r_s}{f k_i (\text{TVL}_i + \epsilon_v w_i)} \right)^{\frac{1}{\epsilon_v - 2}} - \text{TVL}_i, \quad i = 1, \ldots, n,
    \end{equation}
and $w_0 = W - \sum_{i=1}^{n} w_i,$
is the solution to the optimization problem (\ref{eq:optimizationProblem_eps}) and maximizes LP's earnings. For $w_i$ small with respect to TVL$_i$, the above equation can be approximated:
    \begin{equation}\label{eq:optimal_allocation_approx}
w_i \approx \left( \frac{r_s}{f k_i \text{TVL}_i } \right)^{\frac{1}{\epsilon_v - 2}} - \text{TVL}_i, \quad i = 1, \ldots, n.
    \end{equation}
\end{proposition}
\begin{proof}
The proof is included in Appendix~\ref{sec:proofs}.
%The statement is proved as in Proposition~\ref{prop: optallocation} using the Lagrangian $\mathcal{L}(\mathbf{w}, \lambda) = r_s w_0 + %\sum_{i=1}^{n} {f \cdot k_i \cdot (\text{TVL}_i + w_i)^{\epsilon_v - 1}} w_i + \lambda \left(W - \sum_{i=0}^{n} w_i\right).$ \qed
\end{proof}
The optimal allocation can be numerically calculated using \eqref{eq:optimal_allocation}. An approximate result for small $w_i$ is given in \eqref{eq:optimal_allocation_approx}.

\paragraph{Convergence.}

Assume $m$ identical independent LPs invest sequentially into staking and AMM pools. An allocation of capital to a pool increases its TVL and reduces the rate of return. The rate of return for pool $i$ after $j - 1$ LPs have invested is given by:
\begin{equation}\label{eq:conv}
r_i^j(w_i^j) = f \cdot k_i \cdot (\text{TVL}_i^{j-1} + w_i^j)^{\epsilon_v - 1}.    
\end{equation}

\begin{proposition}
For $\epsilon_v<1$, the return from each pool converges to the staking rate $r_s$ as $m \to \infty$. That is, $r_i^m(w_i^m) \to r_s \quad \text{as } m \to \infty \quad \text{for all } i = 1, \ldots, n.$
\end{proposition}

\begin{proof}
As LPs invest sequentially, the TVL for each pool increases, which decreases the marginal return $r_i^j(w_i^j)$ in \eqref{eq:conv}.
When $r_i^j(w_i^j) \leq r_s$, no rational LP will add more liquidity to that pool, preferring staking instead. Thus, the return from all AMM pools stabilizes at $r_s$ in the limit $m \to \infty$.
\end{proof}
Note that if $\epsilon_v>1$, increasing TVL increases the return instead of reducing it. As more LPs enter, trading volume scales more than liquidity, leading to increasing LP rewards rather than convergence to $r_s$. This creates a self-reinforcing effect, where AMM pools can keep attracting liquidity as returns remain high, contradicting the equilibrium assumption. This implies that LPs should always prefer AMM liquidity over staking when $\epsilon_v>1$, as staking will always yield lower returns.
For $\epsilon_v=1$, the rates of return are constant for each pool.

\begin{comment}
\subsection{Impact on the Trading Fee \textcolor{red}{subsection needed?}}
Slippage (Price impact), along with gas fees, block slippage, and LP fees, is one of the major components of the trading fee for traders at AMM. For the CPMM, the price impact is calculated as $ \frac{\Delta x}{x} $. Assuming equilibrium between the AMM pool and staking, we obtain that $\text{TVL} = f \cdot \text{Vol} / r_s$. Further, as $2x = \text{TVL}$ for CPMM, we can express the slippage as in CPMM as
\[
\rho = \frac{\Delta x}{x} = \frac{2r_s}{f \cdot \text{Vol}} \cdot \Delta x
\]
This can be observed because the price impact of trade is proportional to the stake rate and reversely proportional to LP fee ($f$) and volume in the pool.
\end{comment}

%% file: sections/04EmpiricalAnalysis.tex
%%%%%%%%%%%%%%%%%%%%%%%%%%%%%%%%%%%%%%%%%%%%%%%%%%%%%%%%%%%%%%%%%%%%%%%%%%%%%%%%%%%%%%%%%%%%%%%%%%%%%%%%%%%%%%%%%%
\section{Empirical Analysis}
    \label{sec:dataset}

\begin{figure*}[!tbp]
  \centering
  % First row of the 2x2 grid
  \begin{subfigure}[t]{0.48\linewidth}
    \centering
    \includegraphics[width=\linewidth]{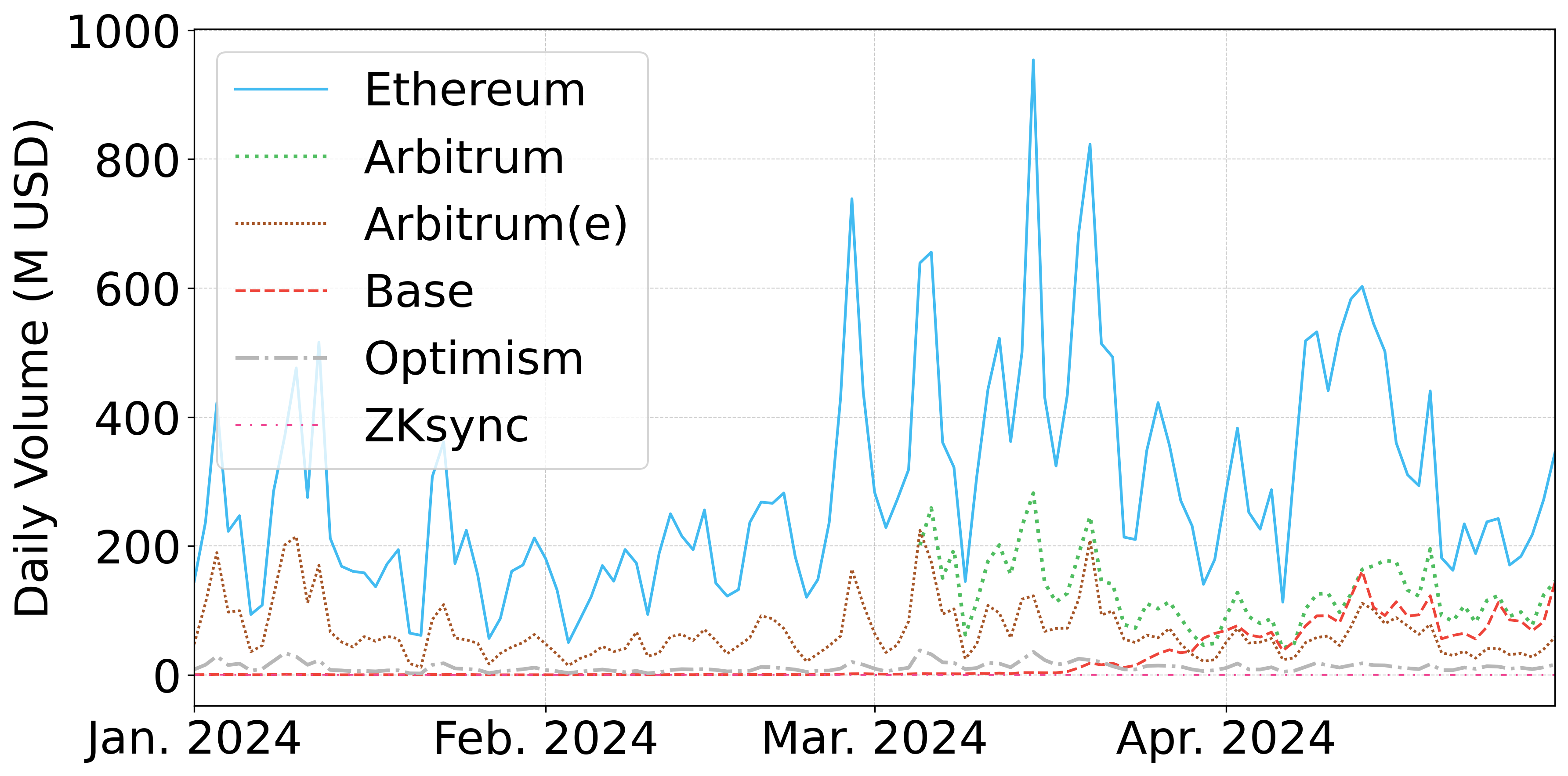}
    \caption{Daily Trading Volume.}
    \label{fig:volume}
  \end{subfigure}
  \hfill
  \begin{subfigure}[t]{0.48\linewidth}
    \centering
    \includegraphics[width=\linewidth]{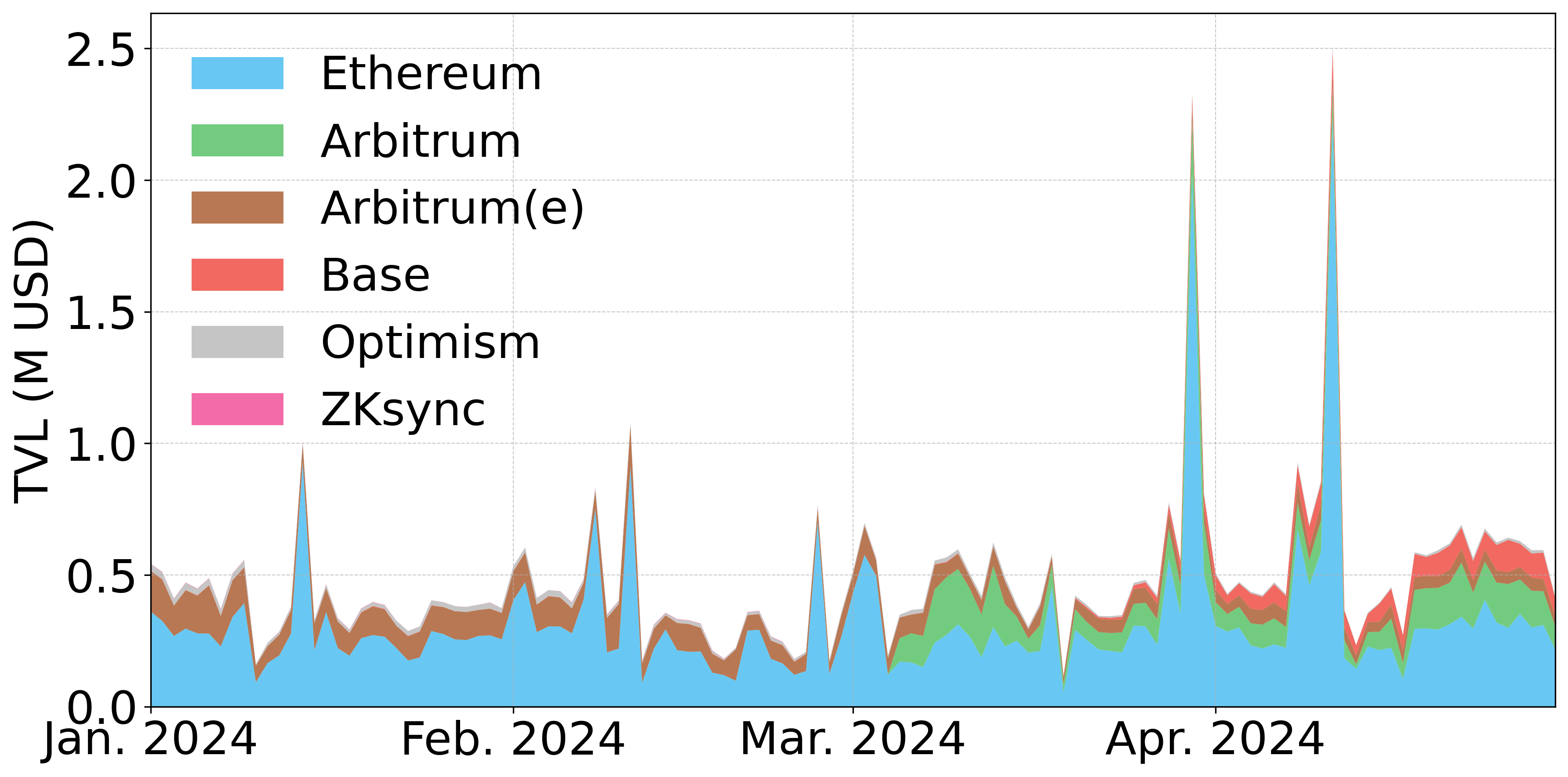}
    \caption{TVL in the Current Tick.}
    \label{fig:TVL}
  \end{subfigure}

  \vspace{0.5cm} % Add vertical spacing between rows

  % Second row of the 2x2 grid
  \begin{subfigure}[t]{0.48\linewidth}
    \centering
    \includegraphics[width=\linewidth]{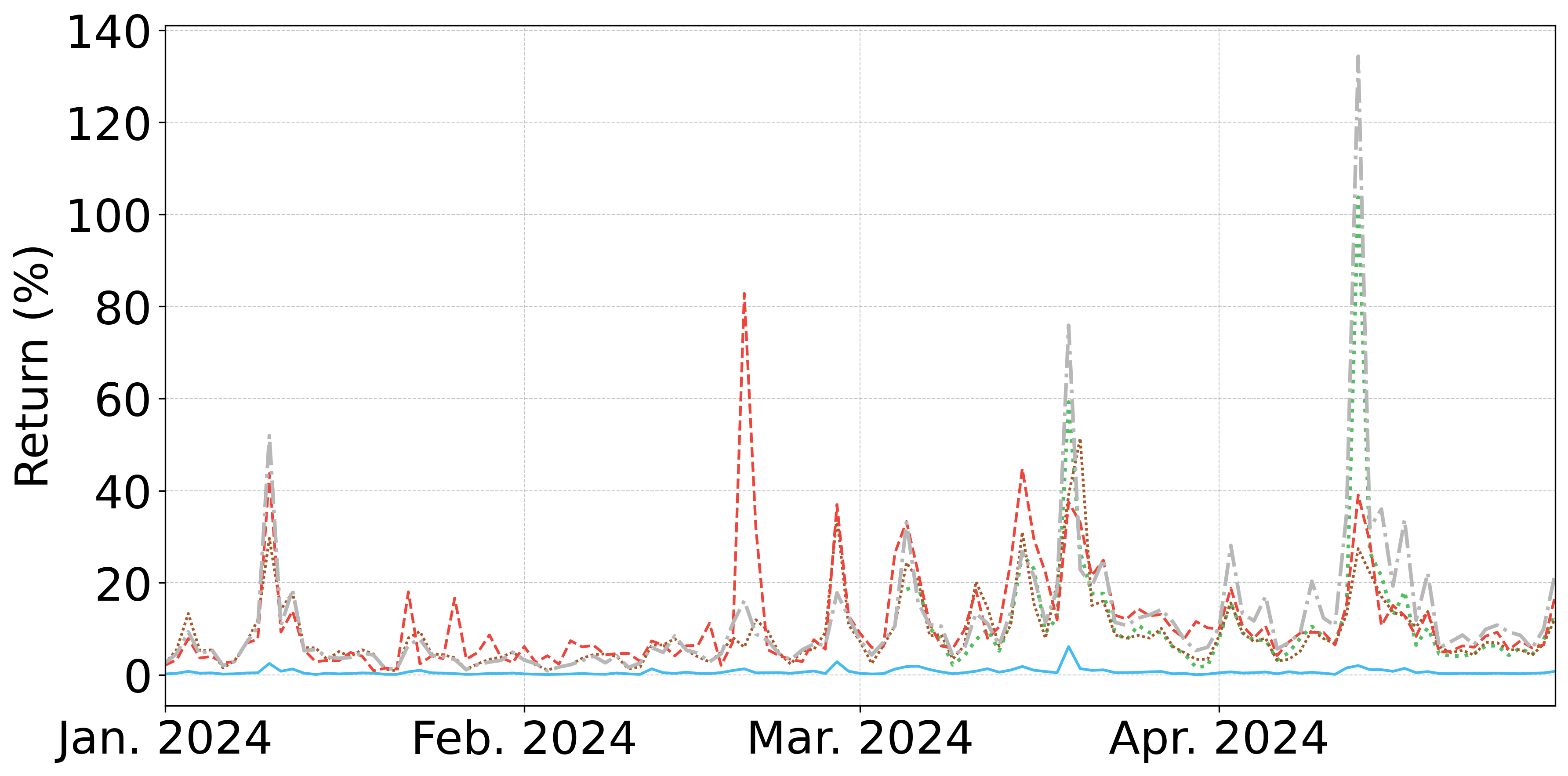}
    \caption{LP Return.}
    \label{fig:APY}
  \end{subfigure}
  \hfill
  \begin{subfigure}[t]{0.48\linewidth}
    \centering
    \includegraphics[width=\linewidth]{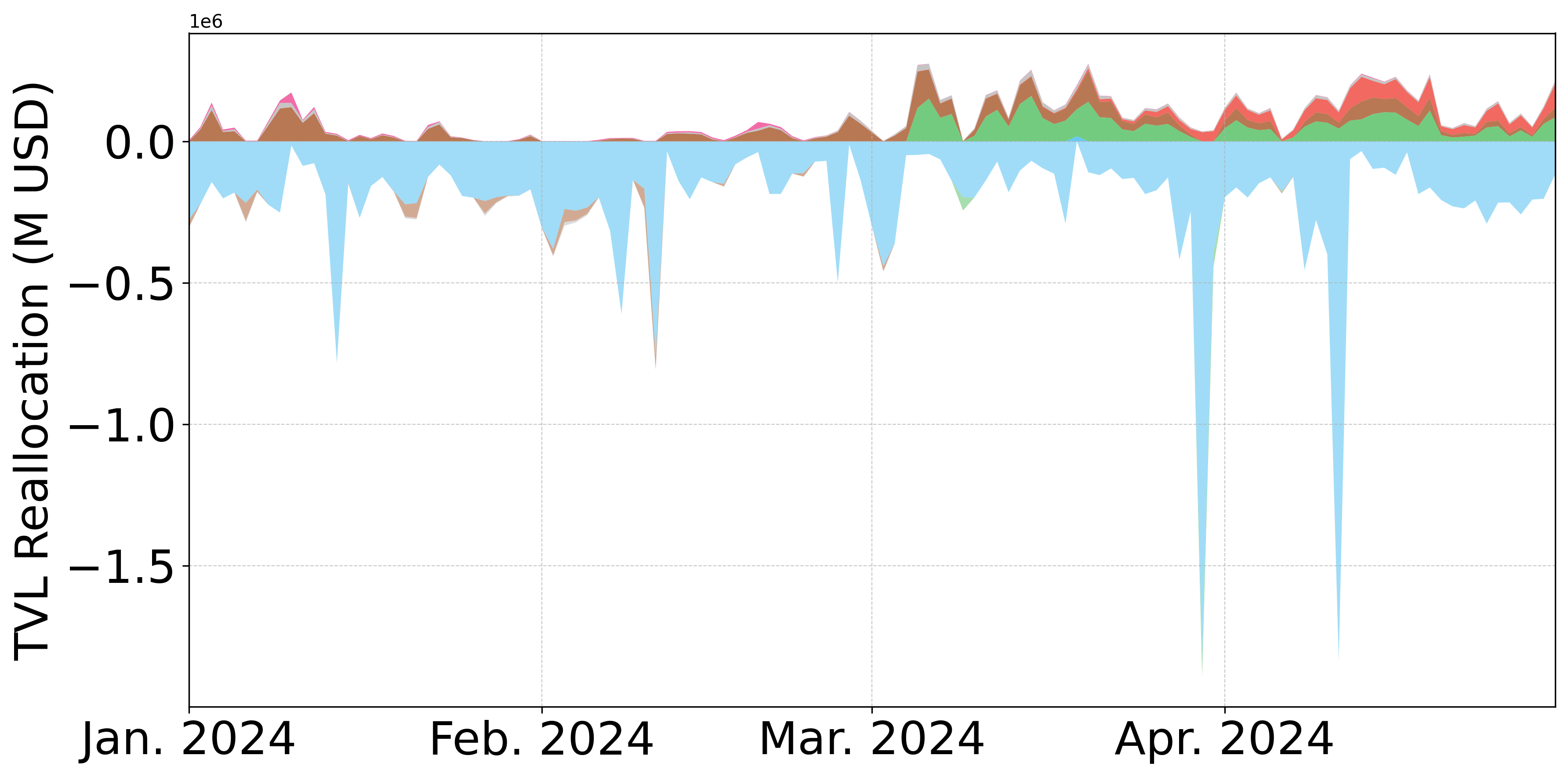}
    \caption{Liquidity Reallocation.}
    \label{fig:allocation}
  \end{subfigure}

  \caption{Overview of WETH-USDC liquidity pools on Uniswap (v3) across Ethereum and its rollups: trading volume, TVL in the current tick and LP returns. The last chart illustrates the reallocation of liquidity required to attain equilibrium. Liquidity that is not allocated to any pool is directed towards staking.}
  \label{fig:metrics_by_chain}
\end{figure*}

\parai{Data.}
For our empirical analysis, we focus on the WETH-USDC pools on Uniswap v3, as they exhibit the highest trading volume and TVL. Additionally, we always examine the WETH-USDC pool on Ethereum and across each L2 to ensure that all LPs are exposed to the same market risks. We collect on-chain data for every swap in these pools, covering the period from 2023 to the first half of 2024. The analysis includes Ethereum and its EVM-compatible L2s, specifically optimistic rollups such as Arbitrum, Optimism, and Base, as well as a ZK-rollups --- ZKsync. For Arbitrum, we analyze pools with native USDC and bridged USDC.e against WETH. The Arbitrum pool with bridged USDC.e we further denote as Arbitrum (e). LP fee in these pools is 5bps, with the exception of ZKsync pool with 20bps. The trading volume in these pools is shown in Figure~\ref{fig:volume}.

\parai{Methodology.}
For each day and each pool, we find the last swap and, based on the liquidity and current tick, we calculate TVL in the current tick, shown in Figure~\ref{fig:TVL}. Then, we calculate the returns of the AMM pool based on the TVL in the current tick using Equation~\eqref{eq:return_tick} and the optimal allocation using Equation~\eqref{eq:prop1}. We also assume that the LP provides liquidity in ticks around 12\% of the current spot price and, based on the finding of Adams~\cite{adams2024layer2layer2}, the liquidity concentration is 75\% higher on L2s compared to Ethereum. The annualized return of such LP position is depicted in Figure~\ref{fig:APY}.

%%%%%%%%%%%%%%%%%%%%%%%%%%%%%%%%%%%%%%%%%%%%%%%%%%%%%%%%%%%%%%%%%%%%%%%%%%%%%%%%%%%%%%%%%%%%%%
\parai{Empirical Calibration.} 
To model the relationship between trading volume and TVL in liquidity pools, we use the function defined in Equation~\eqref{eq: Vol}.
To estimate the parameters $k$ and $\epsilon_v$, we apply a logarithmic transformation to linearize the equation:
\begin{equation}
\log(\text{Vol}) = \log(k) + \epsilon_v \log(\text{TVL}),
\end{equation}
which can be expressed as a linear regression model 
$y = a + b x$
where \( y = \log(\text{Vol}) \), \( x = \log(\text{TVL}) \), \( a = \log(k) \), and \( b = \epsilon_v \).
%We perform a linear regression analysis on historical on-chain data for each liquidity pool.
%The slope of the regression line provides an estimate of \( \epsilon_v \), while the intercept gives \( \log(k) \), from which \( k \) can be obtained by exponentiating the intercept: $k = e^a$.
%The elasticity parameter \( \epsilon_v \) indicates how sensitive trading volume is to changes in TVL. An elasticity greater than 1 (\( \epsilon_v > 1 \)) suggests a highly elastic volume, where trading activity increases significantly as liquidity grows. Conversely, an elasticity of 1 or less (\( \epsilon_v \leq 1 \)) indicates inelastic volume, where trading volume grows slower than TVL.

\begin{table}
\centering
\setlength{\tabcolsep}{5pt}
\begin{tabular}{lrrrrr}
\toprule
    \textbf{Chain} &  \textbf{Elasticity} ($\epsilon_v$) &  \textbf{Scaling constant} ($k$) & \(R^2\) & SE(\(\epsilon_v\)) & SE(\(k\)) \\
\midrule
Ethereum  & -0.121 & 3.05e+09 & 0.012 & 0.101 & 2.091 \\ \hline
Arbitrum  & -0.177 & 3.95e+09 & 0.039 & 0.118 & 2.348 \\ \hline
Arbitrum (e) & -0.143 & 9.19e+08 & 0.010 & 0.128 & 2.488 \\ \hline
Base      &  1.045 & 0.118    & 0.895 & 0.033 & 0.539 \\ \hline
Optimism  & -0.178 & 2.37e+08 & 0.026 & 0.100 & 1.757 \\ \hline
ZKsync    &  0.654 & 4.701    & 0.556 & 0.054 & 0.746 \\
\bottomrule\\
\end{tabular}
\caption{Elasticity analysis across Ethereum and selected rollups. \(\epsilon_v\) represents the elasticity coefficient, \(k\) the scaling constant, \(R^2\) the goodness-of-fit, and SE denotes standard errors for \(\epsilon_v\) and \(k\).}\label{tab:elasticity_scaling}
\end{table}

\subsection{Interpretation of Elasticity and Scaling Constant Results}

Table~\ref{tab:elasticity_scaling} presents the estimated elasticity of volume with respect to virtual TVL and the scaling constant for six different chains. These metrics are derived from a log-log regression model that captures the relationship between TVL and trading volume. The key insights from the results are discussed below.

The elasticity of volume with respect to TVL for Ethereum is estimated to be $-0.12$. This suggests that, contrary to expectations, an increase in TVL on Ethereum is associated with a slight decrease in trading volume. However, the scaling constant indicates a high baseline volume relative to TVL, suggesting that Ethereum has a significant trading volume regardless of TVL changes.
Both Arbitrum and Arbitrum (e) pools exhibit negative elasticity values. These results imply that an increase in TVL leads to a decrease in trading volume on these chains. The scaling constants indicate that Arbitrum has a higher baseline volume compared to Arbitrum (e). Optimism displays a negative elasticity, similar to Arbitrum. The scaling constant for Optimism indicates a lower baseline volume compared to Arbitrum and Ethereum. Yet, the  $R^2$ values are for these chains remain low.

Base shows a positive elasticity value of $1.05$, indicating that a 1\% increase in TVL is associated with a 1.05\% increase in trading volume. The scaling constant for Base is $0.12$, which is significantly lower than other chains, suggesting that Base has a lower baseline trading volume. ZKsync also exhibits a positive elasticity value of $0.65$, indicating that trading volume is positively correlated with TVL on this chain. The scaling constant for zkSync is $4.70$, suggesting a relatively modest baseline volume compared to the other chains.

The elasticity results highlight the variability in the relationship between TVL and trading volume across different chains. While Base and ZKsync --- new blockchains (launched in 2023) with new deployments of Uniswap (v3) show positive elasticities, suggesting that TVL growth drives volume, Ethereum, Arbitrum, and Optimism --- well established and older chains - show negative elasticities with low $R^2$ values indicating that the relation between TVL and Volume cannot be easily established.  This can be attributed to differences in user behavior, liquidity distribution, or (lack of) protocol incentives on these chains. Arbitrum and Optimism were launched in 2021, and Ethereum in 2015.

%%%%%%%%%%%%%%%%%%%%%%%%%%%%%%%%%%%%%%%%%%%%%%%%%%%%%%%%%%%%%%%%%%%%%%%%%%%%%%%%%%%%%%%%%%%%%%
\subsection{Interpretation of Optimal Allocation Results}
%\parab{Results.}
Table \ref{tab:chain_data} presents the calculated results for the new LP willing to allocate \$212k liquidity among the pools on Ethereum and its L2s on 30th April 2024, assuming 3.47\% staking rate of ETH. If the new LP would possess more liquidity, the additional liquidity should be allocated to staking. If the LP would possess less liquidity, it should be allocated to the pools with the highest rewards first. The first observation is that the AMM pool on Ethereum is oversubscribed with a return rate lower than the staking rate. Thus, the current LPs would yield higher returns by reallocating their capital to L2s, or to staking. Second, the highest return presents the pool on Optimism --- 22.12\%, however, given the trading volume on Optimism, the highest capital allocation in optimal LP strategy is achieved for the pools on Arbitrum and BASE, approx. 80k USD. The returns of LP after allocating capital to the pools on L2s are presented in the last column of Table~\ref{tab:chain_data} and ranges between 6 and 8.5\%, significantly lower compared to the current returns of 12 to 22\%.

\begin{table*}[tp]
\centering
\setlength{\tabcolsep}{5pt}
\resizebox{\textwidth}{!}{%
\begin{tabular}{lrrrrr}
\toprule
\textbf{Chain} & \textbf{TVL} & \textbf{Daily Volume} & \textbf{Return (\%)} & \textbf{Allocation} & \textbf{LP Return (\%)} \\ \midrule
Ethereum   & \num{219361.63}  & \num{344883212.09}  & \num{3.03}  & -  & - \\ \hline
Arbitrum   & \num{87534.74}   & \num{145128466.78}  & \num{13.43} & \num{84699.60}   & \num{6.62} \\ \hline
Arbitrum(e)  & \num{38322.85}   & \num{58601891.49}   & \num{12.33} & \num{33913.61}   & \num{6.36} \\ \hline
Base       & \num{68627.34}   & \num{142939078.94}  & \num{17.16} & \num{83973.61}   & \num{7.38} \\ \hline
Optimism   & \num{6147.50}    & \num{16156525.13}   & \num{22.12} & \num{9372.02}    & \num{8.24} \\ \hline
ZKsync     & \num{103.96}     & \num{26430.69}      & \num{16.72} & \num{124.26}     & \num{7.30} \\ 
\bottomrule\\
\end{tabular}}
\caption{Calculations results for the new LP to the pools on Ethereum and its L2s on 30th April 2024: TVL in the current tick after the last swap of the day, the daily traded volume, and the annualized return of the pools. Further, the optimal allocation for the new LP assuming the staking rate of ETH equal to 3.47\%. The last column presents the returns for the LP, based on the given allocations. There is no allocation into the pool on Ethereum, as the return is below staking rate.}
\label{tab:chain_data}
\end{table*}

%The following data were extracted from event logs recorded on the blockchain by Automated Market Makers (AMMs): the reserves of tokens within the liquidity pool (token0 and token1 balances ascertained via the Sync method), the quantity of LP-tokens in circulation (determined by invoking the TokenSupply function on a historical node), and trading volumes (quantities of swap tokens derived from the Swap method).

\begin{figure}[t]
\centering
\includegraphics[width=\linewidth]{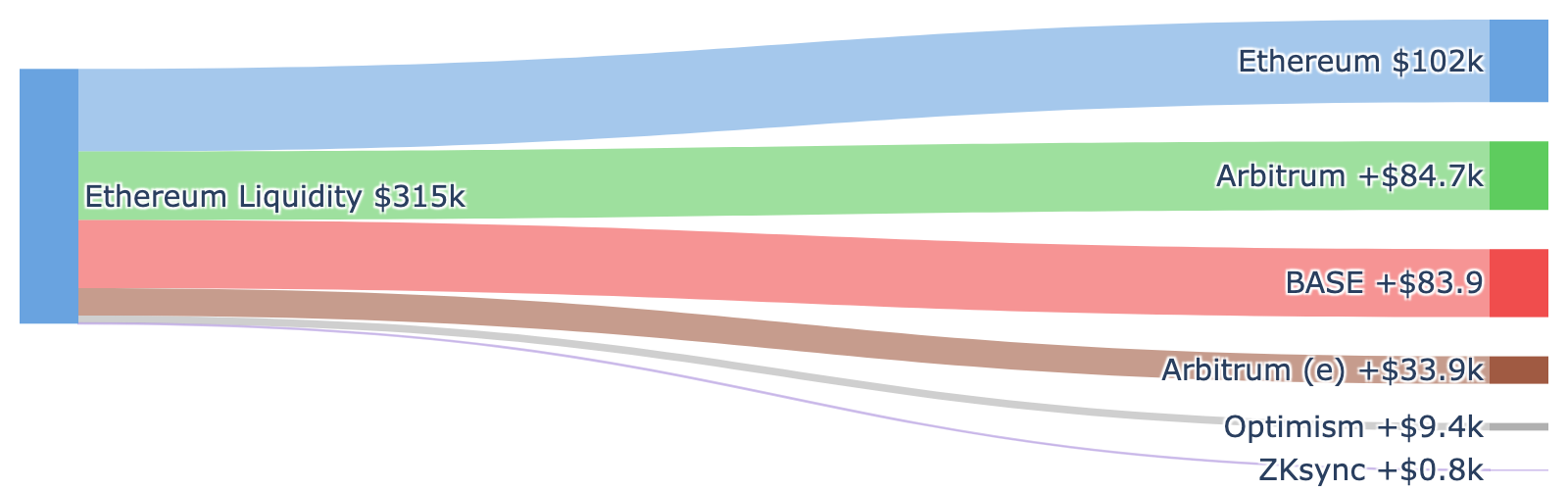}
\caption{The optimal redistribution of current liquidity among blockchains and with the staking rate to attain equilibrium as of 30th April. The source liquidity is the current tick within the Ethereum Uniswap (v3) WETH-USCD pool (\$315k) and corresponding target pools are: Ethereum (\$102k), Base (\$83.9k), Arbitrum (\$84.7k), Arbitrum(e) (\$33.9k), Optimism, (\$9.4k), and ZKsync (\$0.8k).}
\label{fig:30April}
\end{figure}

The optimal allocation formula can be is applied to identify the optimal allocation of current liquidity that is deployed to Uniswap (v3) pools across Ethereum and rollups. This is achieved by allowing the negative allocations and assuming a initial LP wealth of 0.  Figure~\ref{fig:volume} illustrates the re-allocation of liquidity among pools to attain equilibrium. The results indicate that the present allocation of liquidity to Ethereum is excessively high and should be redistributed to pools on alternative chains and towards staking. 

More detailed analysis is depicted in Figure~\ref{fig:30April} for 30th April 2024,. It presents that over 2/3 of current capital allocated to the current tick of the pool at Ethereum should be reallocated, mostly to pools on Arbitrum and Base. This corresponds to Figure~\ref{fig:volume} depicting trading volumes, which are only 2-3 times higher at Ethereum than Arbitrum and Base, very low current LP returns for liquidity provision at Ethereum (below staking rated) as depicted in Figure~\ref{fig:volume}.

%% file: sections/08Discussion.tex
\section{Discussion}
\label{sec:discussion}
A key question arising from our analysis is why LPs do not reallocate their capital to AMMs on L2s or to staking. Since our methodology ensures that all selected pools share the same market risk exposure, LPs’ reluctance toward L2s likely stems from security concerns, inertia, or a preference for Ethereum’s infrastructure. Potential security concerns include risks associated with centralized sequencers, the upgradeability of L2 protocols, and bridging vulnerabilities. Similar risk-related preferences have been observed among AMM traders on L2s~\cite{Chemaya2022Preferecnes}.

\parai{Cost of Bridging.}
The cost of transferring tokens between Ethereum and L2s is minimal—less than 1 basis point (0.01\%) of the trade volume when intent-based interoperability protocols~\cite{chitra2024Intent}, such as Across Protocol~\cite{2025Across}, are used. For instance, in the case of a \$100k transfer (as illustrated in Figure~\ref{fig:30April}), the cost would not exceed \$5 per rollup, with transaction times ranging from just a few seconds to a maximum of 15 minutes~\cite{2025Across}.
Furthermore, gas fees on Ethereum and L2s do not pose a significant barrier to token transfers, as these costs remain low—below \$2 on Ethereum and approximately \$0.10 on L2s~\cite{2025L2Fees}.

\parai{Time of Bridging.}
The time of transfer during which the tokens cannot earn LP fees lasts only seconds (up to 15 minutes), as bridges and intent-based interoperability protocols absorb the risk of L2 transaction not reaching the full finality on Ethereum (\textit{hard finality}), but assume trust in the centralized sequencers and execute the transfers once the transaction in included L2 block by the sequencer (\textit{soft finality})~\cite{yee2022shades}. Especially for optimistic rollups this approach reduces the finality time from 7 days of hard finality to 0.2-2.0s of soft finality.

\parai{Future Research.}
%Lastly, as L2s are cheaper and faster chains compared to Ethereum, LPs can adjust their CLMM positions more often, resulting in higher liquidity concentration on L2s~\cite{adams2024layer2layer2}, and higher returns from liquidity provisions. 
LP rewards should not only compensate for impermanent loss and the opportunity cost of not staking capital (loss-versus-staking)~\cite{gogol2024liquid} but should also reflect the volatility of the traded assets. AMM pools with more volatile tokens should yield higher returns relative to staking than pools with less volatile pairs.
This study focuses on the highly liquid WETH-USDC pool on Uniswap v3. Future research could expand the analysis to other pools, particularly those with lower liquidity or higher volatility (e.g., memecoin pools), where loss-versus-rebalancing (LVR)~\cite{milionis2023automated} and impermanent loss may significantly influence LP allocation decisions.
A sensitivity analysis that incorporates variations in key parameters, such as TVL, trading volume, and stake rates, or an analysis of other relationships between TVL and trading volume, could provide further insight into the behavior of LP.

%% file: sections/09Conclusions.tex
%%%%%%%%%%%%%%%%%%%%%%%%%%%%%%%%%%%%%%%%%%%%%%%%%%%%%%%%%%%%%%%%%%%%%%%%%%%%%%%%%%%%%%%%%%%%%%%%%%%%%%%%%%%%%%%%%%
\section{Conclusions}
    \label{sec:conclusions}

The question of whether L2 blockchains cause the liquidity fragmentation of Ethereum has been the subject of extensive debate. Our research indicates that such fragmentation is currently not occurring, but, it could develop in the future—particularly if LPs become aware of potentially higher returns available on L2s.
Using Lagrangian optimization, we developed a model for the optimal liquidity provisions across AMMs on Ethereum and its L2s with staking as a reference rate. We showed that the returns of the AMM liquidity provision converge to the staking rate. 

In addition, we modeled and measured the elasticity of trading volume with respect to TVL at the WETH-ETH pools at Uniswap (v3) across Ethereum and rollups. Our finding indicate that at the well established blockchains - Ethereum, Arbitrum and Optimism, an increase in TVL is not associated with an increase in trading volume. In contrast, emerging blockchains, Base and ZKsync, exhibit a positive elasticity value, indicating that the volume of trade is positively correlated with TVL on these chains.

Finally, we empirically compared profitability of liquidity provisions to WETH-USDC pools, and observed that the Ethereum pool, compared to corresponding pools on L2s, is oversubscribed and often yields lower returns than staking Ether. Using historical trading volumes and TVLs, we calculated the optimal liquidity allocation for the new LP as well as the optimal capital reallocation across Ethereum and L2s. Specifically, we found that more than 66\% of the Ethereum liquidity could be reallocated to rollups to maximize LP returns and to achieve equilibrium with staking rates.

%% file: sections/98Proofs.tex
\section{Proof of Proposition 3}
\label{sec:proofs}

The Lagrangian function is given by:

\[
\mathcal{L}(\mathbf{w}, \lambda) = r_s w_0 + \sum_{i=1}^{n} f k_i (\text{TVL}_i + w_i)^{\epsilon_v - 1} w_i + \lambda \left( W - \sum_{i=0}^{n} w_i \right).
\]
Derivative w.r.t $w_0$:
\[
\frac{\partial \mathcal{L}}{\partial w_0} = r_s - \lambda = 0.
\]

Solving for \( \lambda \):

\[
\lambda = r_s.
\]

For \( i = 1, \dots, n \):

\[
\frac{\partial \mathcal{L}}{\partial w_i} = f k_i \left[ (\text{TVL}_i + w_i)^{\epsilon_v - 1} + (\epsilon_v - 1) (\text{TVL}_i + w_i)^{\epsilon_v - 2} w_i \right] - \lambda = 0.
\]

Substituting \( \lambda = r_s \):

\[
f k_i \left[ (\text{TVL}_i + w_i)^{\epsilon_v - 1} + (\epsilon_v - 1) (\text{TVL}_i + w_i)^{\epsilon_v - 2} w_i \right] = r_s.
\]

Factor out \( (\text{TVL}_i + w_i)^{\epsilon_v - 2} \):

\[
f k_i (\text{TVL}_i + w_i)^{\epsilon_v - 2} \left[ (\text{TVL}_i + w_i) + (\epsilon_v - 1) w_i \right] = r_s.
\]

Since 

\[
(\text{TVL}_i + w_i) + (\epsilon_v - 1) w_i = \text{TVL}_i + \epsilon_v w_i,
\]

we rewrite the equation as:

\[
f k_i (\text{TVL}_i + w_i)^{\epsilon_v - 2} (\text{TVL}_i + \epsilon_v w_i) = r_s.
\]

Dividing both sides by \( f k_i (\text{TVL}_i + \epsilon_v w_i) \), we get:

\[
(\text{TVL}_i + w_i)^{\epsilon_v - 2} = \frac{r_s}{f k_i (\text{TVL}_i + \epsilon_v w_i)}.
\]

Taking the power \( \frac{1}{\epsilon_v - 2} \) on both sides:

\[
\text{TVL}_i + w_i = \left( \frac{r_s}{f k_i (\text{TVL}_i + \epsilon_v w_i)} \right)^{\frac{1}{\epsilon_v - 2}}.
\]

Solving for \( w_i \):

\[
w_i = \left( \frac{r_s}{f k_i (\text{TVL}_i + \epsilon_v w_i)} \right)^{\frac{1}{\epsilon_v - 2}} - \text{TVL}_i.
\]

Finally, the budget constraint must hold:

\[
W = w_0 + \sum_{i=1}^{n} w_i. \qed
\]